\documentclass[11pt]{article}  

\usepackage{amssymb}
\usepackage{amsmath}
\usepackage{amsthm}
\usepackage{color}
\usepackage{colortbl}
\usepackage{graphicx}
\usepackage{fullpage}
\usepackage{enumerate}
\usepackage{xspace}
\usepackage{algorithm}
\usepackage{caption}
\usepackage{subcaption}
\usepackage{algpseudocode}

\usepackage{lmodern}
\usepackage[T1]{fontenc}

\newtheorem{theorem}{Theorem}
\newtheorem{lemma}[theorem]{Lemma}

\newtheorem{claim}[theorem]{Claim}

\newcommand{\etal}{\emph{et al.}\xspace}
\newcommand{\eps}{\epsilon}
\DeclareMathOperator{\argmax}{argmax}
\DeclareMathOperator{\poly}{poly}

\newcommand{\R}{\mathbb{R}}

\newcommand{\E}{\mathbb{E}}
\newcommand{\Ex}{\mathbb{E}}

\newcommand{\OPT}{\mathrm{OPT}}

\newcommand\cppComment[1]{\texttt{$\langle\langle$ #1 $\rangle\rangle$} \hfill}

\begin{document}
\title{Towards Nearly-linear Time Algorithms for Submodular Maximization with a Matroid Constraint}

\author{
Alina Ene\thanks{Department of Computer Science, Boston University, {\tt aene@bu.edu}.}
\and
Huy L. Nguy\~{\^{e}}n\thanks{College of Computer and Information Science, Northeastern University, {\tt hlnguyen@cs.princeton.edu}.} 
}
\date{}

\maketitle

\thispagestyle{empty}

\begin{abstract}
  We consider fast algorithms for monotone submodular maximization subject to a matroid constraint. We assume that the matroid is given as input in an explicit form, and the goal is to obtain the best possible running times for important matroids.  We develop a new algorithm for a \emph{general matroid constraint} with a $1 - 1/e - \eps$ approximation that achieves a fast running time provided we have a fast data structure for maintaining a maximum weight base in the matroid through a sequence of decrease weight operations. We construct such data structures for graphic matroids and partition matroids, and we obtain the \emph{first algorithms} for these classes of matroids that achieve a nearly-optimal, $1 - 1/e - \eps$ approximation, using a nearly-linear number of function evaluations and arithmetic operations.
\end{abstract}

\section{Introduction}
\label{sec:intro}

In this paper, we consider fast algorithms for monotone submodular maximization subject to a matroid constraint. Submodular maximization is a central problem in combinatorial optimization that captures several problems of interest, such as maximum coverage, facility location, and welfare maximization. The study of this problem dates back to the seminal work of Nemhauser, Wolsey and Fisher from the 1970's~\cite{Nemhauser1978,Nemhauser1978a,Fisher1978}. Nemhauser \etal introduced a very natural Greedy algorithm for the problem that iteratively builds a solution by selecting the item with the largest marginal gain on top of previously selected items, and they showed that this algorithm achieves a $1 - 1/e$ approximation for a cardinality constraint and a $1/2$ approximation for a general matroid constraint. The maximum coverage problem is a special case of monotone submodular maximization with a cardinality constraint and it is $1 - 1/e$ hard to approximate \cite{Feige1998}, and thus the former result is optimal. Therefore the main question that was left open by the work of Nemhauser \etal is whether one can obtain an optimal, $1-1/e$ approximation, for a general matroid constraint.

In a celebrated line of work~\cite{Calinescu2011,Vondrak2008}, Calinescu \etal developed a framework based on continuous optimization and rounding that led to an optimal $1-1/e$ approximation for the problem. The approach is to turn the discrete optimization problem of maximizing a submodular function $f$ subject to a matroid constraint into a continuous optimization problem where the goal is to maximize the multilinear extension $F$ of $f$ (a continuous function that extends $f$) subject to the matroid polytope (a convex polytope whose vertices are the feasible integral solutions). The continuous optimization problem can be solved approximately within a $1-1/e$ factor using a continuous Greedy algorithm \cite{Vondrak2008} and the resulting fractional solution can be rounded to an integral solution without any loss \cite{Ageev2004,Calinescu2011,Chekuri2010}. The resulting algorithm achieves the optimal $1-1/e$ approximation in polynomial time.

Unfortunately, a significant drawback of this approach is that it leads to very high running times. Obtaining fast running times is a fundamental direction both in theory and in practice, due to the numerous applications of submodular maximization in machine learning, data mining, and economics~\cite{LinB10,KrauseSG08,GomesK10,Kempe2003,DughmiRS12}. This direction has received considerable attention~\cite{AzarG12,Filmus2014,Badanidiyuru2014,Mirzasoleiman2015,BFS14,ChekuriJV15}, but it remains a significant challenge for almost all matroid constraints. 

Before discussing these challenges, let us first address the important questions on how the input is represented and how we measure running time. The algorithms in this paper as well as prior work assume that the submodular function is represented as a value oracle that takes as input a set $S$ and returns $f(S)$. For all these algorithms, the number of calls to the value oracle for $f$ dominates the running time of the algorithm (up to a logarithmic factor), and thus we assume for simplicity that each call takes constant time.

The algorithms fall into two categories with respect to how the matroid is represented: the \emph{independence oracle} algorithms assume that the matroid is represented using an oracle that takes as input a set $S$ and returns whether $S$ is feasible (independent); the \emph{representable matroid} algorithms assume that the matroid is given as input in an explicit form.
The representable matroid algorithms can be used for only a subclass of matroids, namely those that can be represented as a linear matroid over vectors in some field\footnote{In a linear matroid, the ground set is a collection of $n$ vectors and a subset of the vectors is feasible (independent) if the vectors are linearly independent.}, but this class includes all the practically-relevant matroids: the uniform, partition, laminar, graphical, and general linear matroids.
The oracle algorithms apply to all matroids, but they are \emph{unlikely to lead to the fastest possible running times:} even an ideal algorithm that makes only $O(k)$ independence calls has a running time that is $\Omega(k^2)$ in the independence oracle model, even if the matroid is a representable matroid such as a partition or a graphic matroid. Thus there have always been parallel lines of research for representable matroids and general matroids.

This work falls in the first category, i.e., we assume that the matroid is given as input in an explicit form, and the goal is to obtain the best possible running times. Note that, although all the representable matroids are linear matroids, it is necessary to consider each class separately, since they have very different running times to even verify if a given solution is feasible: for simple explicit matroids such as a partition or a graphic matroid, checking whether a solution is feasible takes $O(n)$ time, where $n$ is the size of the ground set of the matroid; for general explicit matroids represented using vectors in some field, checking whether a solution is feasible takes $O(k^{\omega})$ time, where $k$ is the rank of the matroid and $\omega$ is the exponent for fast matrix multiplication.

Since in many practical settings only nearly-linear running times are feasible, an important question to address is:
\begin{center}
\emph{For which matroid constraints can we obtain a $1-1/e -\eps$ approximation in nearly-linear time?}
\end{center}
Prior to this work, the only example of such a constraint was a cardinality constraint. For a partition matroid constraint, the fastest running time is $\Omega(n^{3/2})$ in the worst case~\cite{BFS14}. For a graphical matroid constraint, no faster algorithms are known than a general matroid, and the running time is $\Omega(n^2)$. Obtaining a best-possible, nearly-linear running time has been very challenging even for these classes of matroids for the following reasons: 

\begin{itemize}
\item \emph{The continuous optimization is a significant time bottleneck.} The continuous optimization problem of maximizing the multilinear extension subject to the matroid polytope is an integral component in all algorithms that achieve a nearly-optimal approximation guarantee. However, the multilinear extension is expensive to evaluate even approximately. To achieve the nearly-optimal approximation guarantees, the evaluation error needs to be very small and in a lot of cases, the error needs to be $O(n^{-1})$ times the function value. As a result, a single evaluation of the multilinear extension requires $\Omega(n)$ evaluations of $f$. Thus, even a very efficient algorithm with $O(n)$ queries to the multilinear extension would require $\Omega(n^2)$ running time.

\item \emph{Rounding the fractional solution is a significant time bottleneck as well.} Consider a matroid constraint of rank $k$. The fastest known rounding algorithm is the swap rounding, which requires $k$ swap operations: in each operation, the algorithm has two bases $B_1$ and $B_2$ and needs to find $x\in B_1, y\in B_2$ such that $B_1\setminus\{x\}\cup\{y\}$ and $B_2\setminus\{y\}\cup\{x\}$ are bases. The typical implementation is to pick some $x\in B_1$ and try all $y$ in $B_2$, which requires us to check independence for $k$ solutions. Thus, overall, the rounding algorithm checks independence for $\Omega(k^2)$ solutions. Furthermore, each feasibility check takes $\Omega(k)$ time just to read the input. Thus a generic rounding for a matroid takes $\Omega(k^3)$ time. 
\end{itemize}
Thus, in order to achieve a fast overall running time, one needs fast algorithms for both the continuous optimization and the rounding. In this work, we provide such algorithms for partition and graphic matroids, and we obtain the first algorithms with nearly-linear running times. At the heart of our approach is a general, nearly-linear time reduction that reduces the submodular maximization problem to two data structure problems: maintain an approximately maximum weight base in the matroid through a sequence of decrease-weight operations, and maintain an independent set in the matroid that allows us to check whether an element can be feasibly added. This reduction applies to any representable matroid, and thus it opens the possibility of obtaining faster running times for other classes of matroids. 

\subsection{Our contributions}

We now give a more precise description of our contributions. We develop a new algorithm for maximizing the multilinear extension subject to a \emph{general matroid constraint} with a $1 - 1/e - \eps$ approximation that achieves a fast running time provided we have fast data structures with the following functionality:

\begin{itemize}
\item \emph{A maximum weight base data structure:} each element has a weight, and the goal is to maintain an approximately maximum weight base in the matroid through a sequence of operations, where each operation can only decrease the weight of a single element;
\item \emph{An independent set data structure} that maintains an independent set in the matroid and supports two operations: add an element to the independent set, and check whether an element can be added to the independent set while maintaining independence.
\end{itemize}

\begin{theorem}
\label{thm:general-algo}
Let $f$ be a monotone submodular function and let $\mathcal{M}$ be a matroid on a ground set of size $n$. Let $F$ be the multilinear extension of $f$ and $P(\mathcal{M})$ be the matroid polytope of $\mathcal{M}$. Suppose that we have a data structure for maintaining a maximum weight base and independent set as described above. There is an algorithm for the problem $\max_{x \in P(\mathcal{M})} F(x)$ that achieves a $1 - 1/e - \eps$ approximation using $O(n \; \mathrm{poly}(\log{n}, 1/\eps))$ calls to the value oracle for $f$, data structure operations, and additional arithmetic operations.
\end{theorem}

Using our continuous optimization algorithm and additional results, we obtain the first nearly-linear time algorithms for both the discrete and continuous problem with a graphic and a partition matroid constraint. In the graphic matroid case, the maximum weight base data structure is a dynamic maximum weight spanning tree (MST) data structure and the independent data structure is a dynamic connectivity data structure (e.g., union-find), and we can use existing data structures that guarantee a poly-logarithmic amortized time per operation \cite{HLT01,galler1964improved,tarjan1975efficiency}. For a partition matroid, we provide in this paper data structures with a constant amortized time per operation. We also address the rounding step and provide a nearly-linear time algorithm for rounding a fractional solution in a graphic matroid. A nearly-linear time rounding algorithm for a partition matroid was provided in \cite{BFS14}.

\begin{theorem}
\label{thm:partition}
  There is an algorithm for maximizing a monotone submodular function subject to a generalized partition matroid constraint that achieves a $1 - 1/e - \eps$ approximation using $O(n \poly(1/\eps, \log n))$ function evaluations and arithmetic operations.
\end{theorem}

\begin{theorem}
\label{thm:graphic}
  There is an algorithm for maximizing a monotone submodular function subject to a graphic matroid constraint that achieves a $1 - 1/e - \eps$ approximation using $O(n \poly(1/\eps, \log n))$ function evaluations and arithmetic operations.
\end{theorem}
 
Previously, the best running time for a partition matroid was $\Omega(n^{3/2} \poly(1/\eps, \log{n}))$ in the worst case~\cite{BFS14}. The previous best running time for a graphic matroid is the same as the general matroid case, which is $\Omega(n^2 \poly(1/\eps, \log{n}))$ in the worst case~\cite{Badanidiyuru2014}.

The submodular maximization problem with a partition matroid constraint also captures the \emph{submodular welfare maximization} problem: We have a set $\mathcal{N}$ of $m$ items and $k$ players, and each player $i$ has a valuation function $v_i: 2^{\mathcal{N}} \rightarrow \R_{\geq 0}$ that is submodular and monotone. The goal is to find a partition of the items into $k$ sets $S_1, \dots, S_k$ that maximizes $\sum_{i = 1}^k v_i(S_i)$. We can reduce the submodular welfare problem to the submodular maximization problem with a partition matroid constraint as follows \cite{Vondrak2008}. We make $k$ copies of each item, one for each player. We introduce a partition matroid constraint to ensure that we select at most one copy of each item, and a submodular function $f(S) = \sum_{i = 1}^k v_i(S_i)$, where $S_i$ is the set of items whose $i$-th copy is in $S$.

Using this reduction, we obtain a fast algorithm for the submodular welfare maximization problem as well. The size of the ground set of the resulting objective is $m \cdot k$, where $m$ is the number of items and $k$ is the number of players, and thus the algorithm is nearly-linear in $n = m \cdot k$. Note that, even in the case of modular valuation functions, $n$ is the size of the input to the submodular welfare problem, since we need to specify the valuation of every player for every item.

\begin{theorem}
\label{thm:submodular-welfare}
There is a $1 - 1/e - \eps$ approximation algorithm for submodular welfare maximization using $O(n \poly(1/\eps, \log{n}))$ function evaluations and arithmetic operations.
\end{theorem}

We conclude with a formal statement of the contributions made in this paper on which the results above are based.

\begin{theorem}
\label{thm:partition-data-structure}
There is a dynamic data structure for maintaining a maximum weight base in a partition matroid through a sequence of decrease weight operations with an $O(1)$ amortized time per operation.
\end{theorem}

\begin{theorem}
\label{thm:graphic-swap-rounding}
There is a randomized algorithm based on swap rounding for the graphic matroid polytope that takes as input a point $x$ represented as a convex combination of bases and rounds it to an integral solution $S$ such that $\E[f(S)] \geq F(x)$. The running time of the algorithm is $O(nt \log^2{n})$, where $t$ is the number of bases in the convex combination of $x$.
\end{theorem}

\subsection{Technical overview}
\label{sec:techniques}

The starting point of our approach is the work~\cite{BFS14}. They observed that the running time of the continuous algorithm using the multilinear extension of~\cite{Badanidiyuru2014} depends on the value of the maximum weight base when the value is measured in the modular approximation $f'(S) = \sum_{e\in S} f(e)$. It is clear that this approximation is at least the original function and it can be much larger. They observed that the running time is proportional to the ratio between the maximum weight base when weights are measured using the modular approximation compared with the optimal solution when weights are measured using the original function. On the other hand, in the greedy algorithm, the gain in every greedy step is proportional to the maximum weight base when weights are measured using the modular approximation. Thus, the discrete greedy algorithm makes fast progress precisely when the continuous algorithm is slow and vice versa. Therefore, one can start with the discrete greedy algorithm and switch to the continuous algorithm when the maximum weight solution is small even when weights are measured using the modular approximation.

Our algorithm consists of two key components: (1) a fast dynamic data structure for maintaining an approximate maximum weight base through a sequence of greedy steps, and (2) an algorithm that makes only a small number of queries to the data structure. Even if fast dynamic data structures are available, previous algorithms including that of \cite{BFS14} cannot achieve a fast time, since they require $\Omega(nk)$ queries to the data structure: the algorithm of \cite{BFS14} maintains the marginal gain for every element in the current base and it updates them after each greedy step; since each greedy step changes the marginal gain of every element in the base, this approach necessitates $\Omega(k)$ data structure queries per greedy step.

Our new approach uses random sampling to ensure that the number of queries to the data structure is nearly-linear. After each greedy step, our algorithm randomly samples elements from the base to check and update the marginal gains. Because of sampling, it can only ensure that at least $1/2$ of the elements in every value range have good estimates of their values. However, this is sufficient for maintaining an approximate maximum weight base. The benefit is that the running time becomes much faster: the number of checks that do not result in updates is small and if we make sure that an update only happens when the marginal gain change by a factor $1-\eps$ then the total number of updates is at most $O(n\log n/\eps)$. Thus we obtain an algorithm with only a nearly-linear number of data structure queries and additional running time for \emph{any matroid constraint}.

Our approach reduces the algorithmic problem to the data structure problem of maintaining an approximate maximum weight base through a sequence of value updates. In fact, the updates are only decrement in the values and thus can utilize even the decremental data structures as opposed to fully dynamic ones. In the case of a partition matroid constraint, one can develop a simple ad-hoc solution. In the case of a graphic matroid, one can use classical data structures for maintaining minimum spanning trees~\cite{HLT01}.

In both cases, fast rounding algorithms are also needed. The work~\cite{BFS14} gives an algorithm for the partition matroid. We give an algorithm for the graphic matroid based on swap rounding and classical dynamic graph data structures. To obtain fast running time, in each rounding step, instead of swapping a generic pair, we choose a pair involving a leaf of the spanning tree.

\subsection{Basic definitions and notation}

{\bf Submodular functions.} Let $f: 2^V \rightarrow \mathbb{R}_+$ be a set function on a finite ground set $V$ of size $n := |V|$. The function is \emph{submodular} if $f(A) + f(B) \geq f(A \cap B) + f(A \cup B)$ for all subsets $A, B \subseteq V$. The function is \emph{monotone} if $f(A) \leq f(B)$ for all subsets $A \subseteq B \subseteq V$. We assume that the function $f$ is given as a value oracle that takes as input any set $S \subseteq V$ and returns $f(S)$. We let $F: [0, 1]^V \rightarrow \mathbb{R}_+$ denote the multilinear extension $f$. For every $x \in [0, 1]^V$, we have
  \[ F(x) = \sum_{S \subseteq V} f(S) \prod_{e \in S} x_e \prod_{e \notin S} (1 - x_e) = \Ex[R(x)],\]
where $R(x)$ is a random set that includes each element $e \in V$ independently with probability $x_e$.

{\bf Matroids.} A matroid $\mathcal{M} = (V, \mathcal{I})$ on a ground set $V$ is a collection $\mathcal{I}$ of subsets of $V$, called independent sets, that satisfy certain properties. In this paper, we consider matroids that are given to the input to the algorithm. Of particular interest are the partition and graphic matroids. A generalized partition matroid is defined as follows. We are given a partition $V_1, V_2, \dots, V_h$ of $V$ into disjoint subsets and budgets $k_1, k_2, \dots, k_h$. A set $S$ is independent ($S \in \mathcal{I}$) if $|S \cap V_i| \leq k_i$ for all $i \in [h]$. We let $k = \sum_{i = 1}^h k_i$ denote the rank of the matroid. A graphic matroid is defined as follows. We are given a graph on $k+1$ vertices and $n$ edges. The independent sets of the matroid are the forests of this graph.

{\bf Additional notation.} Given a set $S \in \mathcal{I}$, we let $f_S$ denote the function $f_S: 2^{V \setminus S} \rightarrow \R_{\geq 0}$ such that $f_S(S') = f(S' \cup S) - f(S)$ for all $S' \subseteq V \setminus S$. We let $\mathcal{M}/S = (V \setminus S, \mathcal{I'})$ denote the matroid obtained by contracting $S$ in $\mathcal{M}$, i.e., $S' \in \mathcal{I'}$ iff $S' \cup S \in \mathcal{I}$. We let $P(\mathcal{M})$ denote the matroid polytope of $\mathcal{M}$: $P(\mathcal{M})$ is the convex hull of the indicator vectors of the bases of $\mathcal{M}$, where a base is an independent set of maximum size.

{\bf Constant factor approximation to $f(\OPT)$.} Our algorithm needs a $O(1)$ approximation to $f(\OPT)$. Such an approximation can be computed very efficiently, see for instance \cite{BFS14}.

\subsection{Paper organization}

In Section~\ref{sec:matroid-algo}, we describe our algorithm for the continuous optimization problem of maximizing the multilinear extension subject to a general matroid constraint, with the properties stated in Theorem~\ref{thm:general-algo}. As discussed in the introduction, our algorithm uses certain data structures to achieve a fast running time. In Section~\ref{sec:data-structures}, we show how to obtain these data structures for partition and graphic matroids; as discussed in the introduction, the independent set data structures are readily available, and we describe the maximum weight base data structures in Section~\ref{sec:data-structures}. The results of Section~\ref{sec:matroid-algo} and \ref{sec:data-structures} give nearly-linear time algorithms for the continuous problem of maximizing the multilinear extension subject to a partition and graphic matroid constraint. To obtain a fast algorithm for the discrete problem, we also need a fast algorithm to round the fractional solution. Buchbinder \etal \cite{BFS14} give a nearly-linear time rounding algorithm for a partition matroid. In Section~\ref{sec:graphic-rounding}, we give a nearly-linear time rounding algorithm for a graphic matroid, and prove Theorem~\ref{thm:graphic-swap-rounding}. These results together give Theorems~\ref{thm:partition} and \ref{thm:graphic}. 

\section{The algorithm for the continuous optimization problem}
\label{sec:matroid-algo}

\begin{algorithm}
\caption{Algorithm for the continuous problem $\max_{x \in P(\mathcal{M})} F(x)$ for a general matroid}
\label{alg:matroid}
\begin{algorithmic}[1]
\Procedure{ContinuousMatroid}{$f, \mathcal{M}, \eps$}
\State $c' = \Theta(1/\eps)$, where the $\Theta$ hides a sufficiently large absolute constant
\State $S = \textproc{LazySamplingGreedy}(f, \mathcal{M}, \eps)$
\State $x = \textproc{ContinuousGreedy}(f_S, \mathcal{M}/S, c', \eps)$ \label{line:continuous-greedy}
\State {\bf return} $\mathbf{1}_S \vee x$ \quad \cppComment{$x \vee y$ is the vector $(x \vee y)_i = \max\{x_i, y_i\}$ for all $i$}
\EndProcedure
\end{algorithmic}
\end{algorithm}

\begin{algorithm}
\begin{algorithmic}[1]
\State $M = \Theta(f(\OPT))$, $c = \Theta(1/\eps)$, $N = 2\ln(k/\eps)/\eps$
\State \cppComment{maintain cached (rounded) marginal values}
\State For each $e \in V$, let $w(e) = (1 - \eps)^N M$ if $f(\{e\}) \leq (1 - \eps)^N M$ and $w(e) = (1 - \eps)^{j - 1} M$ if $f(\{e\}) \in ((1 - \eps)^{j} M, (1 - \eps)^{j - 1} M]$
\State \cppComment{maintain a base $B$ of maximum $w(\cdotp)$ value in a data structure that supports the \Call{UpdateBase}{} operation}
\State $B = \argmax_{S \in \mathcal{I}} \sum_{e \in S} w(e)$
\State \cppComment{maintain a partition of $B$ into buckets}
\State $B^{(j)} = \{e \in B \colon w(e) = (1 - \eps)^{j - 1} M\}$ for each $j \in [N]$
\State $W = \sum_{e \in B} w(e)$
\State \cppComment{main loop}
\State $S = \emptyset$
\For{$t = 1, 2, \dots, k$}{} \label{line:start-main-loop}
  \State Call $\Call{RefreshValues}{}$ 
  \If{$W \leq 4c M$}
    \State {\bf return } $S$ \label{line:terminate}
  \EndIf
  \State Sample an element $e$ uniformly at random from $B$
  \State $S \gets S \cup \{e\}$ \label{line:add-element}
  \State Remove $e$ from the buckets of $B$ for refreshing purpose so that $w(e)$ is now fixed
\EndFor \label{line:end-main-loop}
\end{algorithmic}

\begin{algorithmic}[1]
\Procedure{RefreshValues}{}\Comment{Spot check and update values in all buckets} 
\For{$j = 1$ to $N$} 
  \State $T = 0$
  \While{$T < 4\log_2{n}$}
    \If{$B^{(j)}$ is empty}
      \State Exit the while loop and continue to iteration $j + 1$
    \EndIf
    \State Sample $e$ uniformly at random from $B^{(j)}$
    \State Let $v(e) = f(S \cup \{e\}) - f(S)$ be the current marginal value of $e$
    \If{$v(e) < (1 - \eps)^j M$}
      \State $T = 0$
	\State \Call{UpdateBase}{$e, j, v(e)$}
    \Else
      \State $T \gets T + 1$
    \EndIf
  \EndWhile
\EndFor
\EndProcedure
\end{algorithmic}
\caption{$\textproc{LazySamplingGreedy}(f, \mathcal{M}, \eps)$}
\label{alg:lazy-sampling-greedy}
\end{algorithm}

In this section, we describe and analyze our algorithm for the continuous problem $\max_{x \in P(\mathcal{M})} F(x)$ for a general matroid $\mathcal{M}$, and prove Theorem~\ref{thm:general-algo}. The algorithm is given in Algorithm~\ref{alg:matroid} and it combines the continuous Greedy algorithm of \cite{Badanidiyuru2014} with a discrete Greedy algorithm that we provide in this paper, building on \cite{BFS14}.

\medskip
{\bf The continuous Greedy algorithm.}
The algorithm used on line~\ref{line:continuous-greedy} is the algorithm of \cite{Badanidiyuru2014}. To obtain a fast running time, we use an independent set data structure to maintain the independent sets constructed by the algorithm. The data structure needs to support two operations: add an element to the independent set, and check whether an element can be added to the independent set while maintaining independence. For a partition matroid, such a data structure with $O(1)$ time per operation is trivial to obtain. For a graphic matroid, we can use a union-find data structure~\cite{galler1964improved,tarjan1975efficiency} with a $O(\log^*{k})$ amortized time per operation.

It follows from the work of \cite{Badanidiyuru2014} and \cite{BFS14} that the algorithm has the following properties. 

\begin{lemma}[\cite{Badanidiyuru2014}, \cite{BFS14}]
\label{lem:continuous-greedy}
When run with values $c$ and $\delta$ as input, the $\textproc{ContinuousGreedy}(f, \mathcal{M}, c, \delta)$ algorithm uses $O(n \ln(n/\delta) / \delta^2)$ independent set data structure operations, and $O(cn \ln^2(n/\delta) / \delta^4)$ queries to the value oracle of $f$ and additional arithmetic operations. Moreover, if $\max_{S \in \mathcal{I}} \sum_{e \in S} f(e) \leq c \cdot f(\OPT)$, where $\OPT \in \argmax_{S \in \mathcal{I}} f(S)$, the solution $x$ returned by the algorithm satisfies $F(x) \geq (1 - {1 \over e} - \delta) f(\OPT)$.
\end{lemma}

The discrete Greedy algorithm is given in Algorithm~\ref{alg:lazy-sampling-greedy}. The algorithm works for any matroid constraint for which we can provide a fast data structure for maintaining a maximum weight base. We now describe the properties we require from this data structure. We show how to implement a data structure with these properties in Section~\ref{sec:data-structures} for a graphic matroid and a partition matroid.

\medskip
{\bf The dynamic maximum weight base data structure.}
Algorithm~\ref{alg:lazy-sampling-greedy} makes use of a data structure for maintaining the maximum weight base in the matroid, where each element has a weight and the weights are updated through a sequence of updates that can only decrease the weights. The data structure needs to support the following operation: $\Call{UpdateBase}{}$ decreases the weight of an element and it updates the base to a maximum weight base for the updated weights. The data structures that we provide in Section~\ref{sec:data-structures} for a graphic and a partition matroid support this operation in $O(\mathrm{poly}(\log{k}))$ amortized time.

We note here that the data structure maintains a maximum weight base of the original matroid $\mathcal{M}$, and not the contracted matroid $\mathcal{M}/S$ obtained after picking a set $S$ of elements. This suffices for us, since the discrete Greedy algorithm that we use will not change the weight of an element after it was added to the solution $S$. Due to this invariant, we can show that the maximum weight base $B$ of $\mathcal{M}$ that the data structure maintains has the property that $S \subseteq B$ at all times, and $B \setminus S$ is a maximum weight base in $\mathcal{M} / S$. This follows from the observation that, if an element $e$ is in the maximum weight base $B$ and the only changes to the weights are such that the weight of $e$ remains unchanged and the weights of elements other than $e$ are decreased, then $e$ remains in the new maximum weight base. 

\medskip
{\bf The discrete Greedy algorithm.}
The algorithm (Algorithm~\ref{alg:lazy-sampling-greedy}) is based on the random residual Greedy algorithm of \cite{BFNS14}. The latter algorithm constructs a solution $S$ over $k$ iterations. In each iteration, the algorithm assigns a linear weight to each element that is equal to the marginal gain $f(S \cup \{e\}) - f(S)$ on top of the current solution, and it finds a maximum weight base $B$ in $\mathcal{M}/S$. The algorithm then samples an element of $B$ uniformly at random and adds it to the solution.

As discussed in Section~\ref{sec:techniques}, the key difficulty for obtaining a fast running time is maintaining the maximum weight base. Our algorithm uses the following approach for maintaining an approximate maximum weight base. The algorithm maintains the marginal value of each element (rounded to the next highest power of $(1 - \eps)$), and it updates it in a lazy manner; at every point, $w(e)$ denotes the cached (rounded) marginal value of the element, and it may be stale. 

The algorithm maintains the base $B$ using the data structure discussed above that supports the \Call{UpdateBase}{} operation. Additionally, the elements of $B \setminus S$ are stored into buckets corresponding to geometrically decreasing marginal values. More precisely, there are $N = O(\log(k / \eps)/\eps)$ buckets $B^{(1)}, B^{(2)}, \dots, B^{(N)}$. The $j$-th bucket $B^{(j)}$ contains all of the elements of $B$ with marginal values in the range $((1 - \eps)^j M, (1 - \eps)^{j - 1} M]$, where $M$ is a value such that $f(\OPT) \leq M \leq O(1) f(\OPT)$ (we assume that the algorithm knows such a value $M$, as it can be obtained in nearly-linear time, see e.g. Lemma~{3.3} in \cite{BFS14}). The remaining elements of $B$ that do not appear in any of the $N$ buckets have marginal values at most $(1 - \eps)^N M$; these elements have negligible total marginal gain, and they can be safely ignored.

In order to achieve a fast running time, after each Greedy step, the algorithm uses sampling to (partially) update the base $B$, the cached marginal values, and the buckets. This is achieved by the procedure \Call{RefreshValues}{}, which works as follows. \Call{RefreshValues}{} considers each of the buckets in turn. The algorithm updates the bucket $B^{(j)}$ by spot checking $O(\log{n})$ elements sampled uniformly at random from the bucket. For each sampled element $e$, the algorithm computes its current marginal value and, if it has decreased below the range of its bucket, it moves the element to the correct buckets and call \Call{UpdateBase}{} to maintain the invariant that $B$ is a maximum weight base.

When the algorithm finds an element whose bucket has changed, it resets to $0$ the count for the number of samples taken from the bucket. Thus the algorithm keeps sampling from the bucket until $\Theta(\log{n})$ consecutive sampled elements do not change their bucket. The sampling step ensures that, with high probability, in each bucket at least half of the elements are in the correct bucket. (We remark that, instead of resetting the sample count to $0$, it suffices to decrease the count by $1$, i.e., the count is the total number of samples whose bucket was correct minus the number of samples whose bucket was incorrect. The algorithm then stops when this tally reaches $\Theta(\log{n})$. This leads to an improvement in the running time, but we omit it in favor of a simpler analysis.)

After running \Call{RefreshValues}{}, the algorithm samples an element $e$ uniformly at random from $B \setminus S$ and adds it to $S$. The algorithm then removes $e$ from the buckets; this ensures that the weight of $e$ will remain unchanged for the remainder of the algorithm.

\medskip
Next, we analyze the algorithm and show the approximation and running time guarantees stated in Theorem~\ref{thm:general-algo}.

\subsection{Analysis of the running time}

Here we analyze the running time of the algorithm Algorithm~\ref{alg:matroid}.

\begin{lemma}
The \textproc{ContinuousGreedy} algorithm uses $O(n \ln^2(n/\eps) / \eps^5)$ calls to the value oracle of $f$ and arithmetic operations, and $O(n \ln(n/\eps)/\eps)$ independent set data structure operations.
\end{lemma}
\begin{proof}
Since we run the \Call{ContinuousGreedy}{} algorithm with $c = O(1/\eps)$ and $\delta = \eps$, the lemma follows from Lemma~\ref{lem:continuous-greedy}.
\end{proof}

\begin{lemma}
The \textproc{LazySamplingGreedy} algorithm uses $O(k N \log{n})$ calls to the value oracle of $f$ and maximum weight data structure operations, and $O(n + k N \log{n})$ additional arithmetic operations.
\end{lemma}
\begin{proof}
The running time of \Call{LazySamplingGreedy}{} is dominated by \Call{RefreshValues}{}. Consider the \Call{RefreshValues}{} subroutine. Each iteration of the while loop of \Call{RefreshValues}{} uses $O(1)$ function evaluations and makes one call to the dynamic maximum base data structure. We divide the work performed updating a given bucket $j$ into epochs, where an epoch is a sequence of consecutive iterations  of the while loop starting with $T = 0$ and ending when either $T$ is reset to $0$ on line~{11} --- an intermediate epoch --- or $T$ reaches $4\log_2{n}$ --- a final epoch. Since each epoch has $O(\log{n})$ iterations, it runs in $O(\log{n})$ amortized time and $O(\log{n})$ function evaluations. Every intermediate epoch moves an element $e$ to a lower bucket, and we charge the total work done by the epoch to the element $e$. Since an element $e$ can be charged only $N$ times, we can upper bound the numbers of update base operations and function evaluations of the intermediate epochs across the entire run of \Call{LazySamplingGreedy}{} by $O(n N \log{n})$. The number of final epochs is at most $k N$ and thus their numbers of update base operations and function evaluations is $O(kN \log{n})$.
\end{proof}

\subsection{Analysis of the approximation guarantee}

Here we show that Algorithm~\ref{alg:matroid} achieves a $1 - 1/e - \eps$ approximation.

We first analyze the \textproc{LazySamplingGreedy} algorithm. We start with some convenient definitions. Consider some point in the execution of the \textproc{LazySamplingGreedy} algorithm. Consider a bucket $B^{(j)}$. At this point, each element $e \in B^{(j)}$ is in the correct bucket iff its current marginal value $f(S \cup \{e\}) - f(S)$ lies in the interval $((1 - \eps)^j M, (1 - \eps)^{j - 1} M]$ (its cached marginal value $w(e)$ lies in that interval, but it might be stale). We say that the bucket $B^{(j)}$ is \emph{good} if at least half of the elements in $B^{(j)}$ are in the correct bucket, and we say that the bucket is \emph{bad} otherwise.

The following lemma shows that, with high probability over the random choices of \textproc{RefreshValues}, each run of \textproc{RefreshValues} ensures that every bucket $B^{(j)}$ with $j \in [N]$ is good. 

\begin{lemma}
  Consider an iteration in which \textproc{LazySamplingGreedy} calls \textproc{RefreshValues}. When \textproc{RefreshValues} terminates, the probability that the buckets $\{B^{(j)} \colon j \in [N]\}$ are all good is at least $1 - 1/n^2$.
\end{lemma}
\begin{proof}
  We will show that the probability that a given bucket is bad is at most $5\log n/n^3$; the claim then follows by the union bound, since there are $N \leq n/(5\log n)$ buckets. Consider a bucket $B^{(j)}$, where $j \in [N]$, and suppose that the bucket is bad at the end of \textproc{RefreshValues}. We analyze the probability the bucket is bad because the algorithm runs until iteration $t$, which is the last time the algorithm finds an element in $B^{(j)}$ in the wrong bucket, and for $4\log{n}$ iterations after $t$, it always find elements in the right bucket even though only $1/2$ of $B^{(j)}$ are in the right bucket. Since at most half of the elements of $B^{(j)}$ are in the correct bucket and the samples are independent, this event happens with probability at most $(1/2)^{4\log_2{n}} = 1/n^4$. By the union bound over all choices of $t=1, 2, \ldots, 5n\log n$, the failure probability for bucket $B^{(j)}$ is at most $5\log n/n^3$.
\end{proof}

Since \textproc{LazySamplingGreedy} performs at most $k \leq n$ iterations, it follows by the union bound that all of the buckets $\{B^{(j)} \colon j \in [N]\}$ are all good throughout the algorithm with probability at least $1 - 1/n^2$. For the remainder of the analysis, we condition on this event. Additionally, we fix an event specifying the random choices made by \textproc{RefreshValues} and we implicitly condition all probabilities and expectations on this event. 

Let us now show that $B$ is a suitable approximation for the maximum weight base in $\mathcal{M}/S$ with weights given by the current marginal values $f(S \cup \{e\}) - f(S)$. 
\begin{lemma}
\label{lem:base-value}
  Suppose that every bucket of $B$ is good throughout the algorithm. Let $v(e) = f(S \cup \{e\}) - f(S)$ denote the current marginal values. We have
  \begin{enumerate}[$(1)$]
    \item $w(S') \geq v(S')$ for every $S' \subseteq V$;
    \item $w(B) \geq w(S')$ for every $S' \subseteq V$;
    \item $v(B) \geq {1 - \eps \over 2} \cdot w(B) - {\eps^2 \over k} \cdot M$.
  \end{enumerate}
\end{lemma}
\begin{proof}
  The first property follows from the fact that, by submodularity, the weights $w(\cdot)$ are upper bounds on the marginal values.
  
  The second property follows from the fact that the algorithm maintains the invariant that $B$ is the maximum weight base in $\mathcal{M}/S$ with respect to the weights $w(\cdot)$.
  
  Let us now show the third property. Consider the following partition of $B$ into sets $B_1$, $B_2$, and $B_3$, where: $B_1$ is the set of all elements $e \in B$ such that $e$ is in one of the buckets $\{B^{(j)} \colon j \in [N]\}$ and moreover $e$ is in the correct bucket (note that $(1 - \eps) w(e) \leq v(e) \leq w(e)$ for every $e \in B_1$); $B_2$ is the set of all elements $e \in B$ such that $e$ is in one of the buckets $\{B^{(j)} \colon j \in [N]\}$ but $e$ is not in the correct bucket (i.e., $(1 - \eps)^N M < w(e)$ but $v(e) < (1 - \eps) w(e)$); $B_3$ is the set of all elements $e \in B$ such that $w(e) < (1 - \eps)^N M \leq (\eps/k)^2 M$. 

  Since $|B_3| \leq k$, we have
  \[ w(B_3) \leq |B_3| \cdot \left({\eps\over k}\right)^2 M \leq {\eps^2 \over k} M.\]
  Since all of the buckets are good, it follows that the total $w(\cdot)$ weight of the elements that are in the correct bucket is at least ${1 \over 2} \cdot w(B \setminus B_3)$. Indeed, we have
  \[ w(B_1) = \sum_{j = 1}^N w(B_1 \cap B^{(j)}) = \sum_{j = 1}^N (1 - \eps)^{j - 1}M |B_1 \cap B^{(j)}| \geq \sum_{j = 1}^N (1 - \eps)^{j - 1}M {|B^{(j)}| \over 2} = {w(B \setminus B_3) \over 2}.\] 
  Finally, since $v(e) \geq (1 - \eps) w(e)$ for every $e \in B_1$, we have
  \[ v(B) \geq v(B_1) \geq (1 - \eps) w(B_1) \geq {1 - \eps \over 2} w(B \setminus B_3) \geq {1 - \eps \over 2} w(B) - {\eps^2 \over k} M.\] 
\end{proof}

Now we turn to the analysis of the main for loop of \textproc{LazySamplingGreedy} (lines \ref{line:start-main-loop}--\ref{line:end-main-loop}). Let $T$ be a random variable equal to the number of iterations where the algorithm executes line~\ref{line:add-element}. We define sets $\{S_t \colon t \in \{0, 1, \dots, k\}\}$ and $\{\OPT_t \colon t \in \{0, 1, \dots, k\}\}$ as follows. Let $S_0 = \emptyset$ and $\OPT_0 = \OPT$. Consider an iteration $t \leq T$ and suppose that $S_{t - 1}$ and $\OPT_{t - 1}$ have already been defined and they satisfy $S_{t - 1} \cup \OPT_{t - 1} \in \mathcal{I}$ and $|S_{t-1}|+|\OPT_{t-1}| = k$. Consider a bijection $\pi: B\to \OPT_{t-1}$ so that $\OPT_{t-1}\setminus\{\pi(e)\}\cup\{e\}$ is a base of $\mathcal{M}/S_{t-1}$ for all $e\in B$ (it is well-known that such a bijection exists). Let $e_t$ be the element sampled on line~\ref{line:add-element} and $o_t = \pi(e_t)$. We define $S_t = S_{t - 1} \cup \{e_t\}$ and $\OPT_t = \OPT_{t - 1} \setminus \{o_t\}$. Note that $S_t \cup \OPT_t \in \mathcal{I}$.

In each iteration $t$, the gain in the Greedy solution value is $f(S_t) - f(S_{t - 1})$, and the loss in the optimal solution value is $f(\OPT_{t - 1}) - f(\OPT_t)$ (when we add an element to $S_{t - 1}$, we remove an element from $\OPT_{t - 1}$ so that $S_t \cup \OPT_t$ remains a feasible solution). The following lemma relates the two values in expectation.

\begin{lemma}
\label{lem:greedy-gain}
  For every $t \in [k]$, if all of the buckets $B^{(j)}$ are good, we have
    \[ \Ex[f(S_t) - f(S_{t - 1})] \geq c \cdot \Ex[f(\OPT_{t - 1}) - f(\OPT_t)].\] 
\end{lemma}
\begin{proof}
  Consider an iteration $t \in [k]$. If $t > T$, the inequality is trivially satisfied, since both expectations are equal to $0$. Therefore we may assume that $t \leq T$ and thus $S_t = S_{t - 1}  \cup \{e_t\}$ and $\OPT_t = \OPT_{t - 1} \setminus \{o_t\}$.

  Let us now fix an event $R_{t - 1}$ specifying the random choices for the first $t - 1$ iterations, i.e., the random elements $e_1, \dots, e_{t - 1}$ and $o_1, \dots, o_{t - 1}$. In the following, all the probabilities and expectations are implicitly conditioned on $R_{t - 1}$. Note that, once $R_{t - 1}$ is fixed, $S_{t - 1}$ and $\OPT_{t - 1}$ are deterministic.

  Let us first lower bound $\Ex[f(S_{t - 1} \cup \{e_t\}) - f(S_{t - 1})]$. Let $w_t$, $B_t$, and $W_t$ denote $w$, $B$, and $W$ right after executing \textproc{RefreshValues} in iteration $t$. Note that, since $R_{t - 1}$ and the random choices of \textproc{RefreshValues} are fixed, $w_t$, $B_t$, and $W_t$ are deterministic.

  Recall that all of the buckets of $B_t$ are good, i.e., at least half of the elements of $B_t^{(j)}$ are in the correct bucket, for every $j \in [N]$. Let $B'_t$ be the subset of $B_t$ consisting of all of the elements that are in the correct bucket, and let $B''_t$ be the subset of $B_t$ consisting of all of the elements that are not in any bucket.

  For every $e \in B'_t$, we have
  \[ f(S_{t - 1} \cup \{e\}) - f(S_{t - 1}) \geq (1 - \eps) w_t(e).\]
  For every $e \in B''_t$, we have 
  \[ f(S_{t - 1} \cup \{e\}) - f(S_{t - 1}) \leq w_t(e) = (1 - \eps)^N M \leq (\eps/k)^2 M,\] and therefore
  \[ w_t(B''_t) \leq {\eps^2 \over k} M.\]
  Since all of the buckets are good and $W_t > 4c M$ (the algorithm did not terminate on line~\ref{line:terminate}), we have
  \begin{align*}
    \sum_{e \in B'_t} w_t(e)
    &= \sum_{j = 1}^N w_t(B'_t \cap B^{(j)}_t)
    = \sum_{j = 1}^N (1 - \eps)^{j - 1} M |B'_t \cap B_t^{(j)}|\\
    &\geq \sum_{j = 1}^N (1 - \eps)^{j - 1} M {|B_t^{(j)}| \over 2}
    = {w_t(B_t \setminus B''_t) \over 2}\\
    &\geq {W_t \over 2} - {\eps^2 \over 2k} M
    \geq \left(2c - {\eps^2 \over 2k} \right) M
    \geq \left(2c - {\eps^2 \over 2k} \right) f(\OPT). 
  \end{align*}
  By combining these observations, we obtain
  \begin{align*}
    &\Ex[f(S_{t - 1} \cup \{e_t\}) - f(S_{t - 1})]\\
    &\geq \Ex[f(S_{t - 1} \cup \{e_t\}) - f(S_{t - 1}) | e_t \in B'_t] \Pr[e_t \in B'_t]\\
    &= \Ex[f(S_{t - 1} \cup \{e_t\}) - f(S_{t - 1}) | e_t \in B'_t] \cdot {|B'_t| \over |B_t|}\\
    &\geq (1 - \eps) \Ex[w_t(e_t) | e_t \in B'_t] \cdot {|B'_t| \over |B_t|}\\
    &= (1 - \eps) w_t(B'_t) \cdot {1 \over |B'_t|} \cdot {|B'_t| \over |B_t|}\\
    &\geq {(1 - \eps) \left(2c - {\eps^2 \over 2k} \right) \over |B_t|} f(\OPT)\\
    &\geq {c \over |B_t|} f(\OPT)
  \end{align*}
  Let us now upper bound $\Ex[f(\OPT_{t - 1}) - f(\OPT_t)]$. Recall that $e_t$ is chosen randomly from $B$ and thus, $o_t$ is chosen uniformly at random from $\OPT_{t - 1}$ (since $\pi$ is a bijection). Hence
    \[ \Ex[f(\OPT_{t - 1}) - f(\OPT_{t - 1} \setminus \{o_t\})]\\
    = \sum_{o \in \OPT_{t - 1}} (f(\OPT_{t - 1}) - f(\OPT_{t - 1} \setminus \{o\})) \cdot {1 \over |\OPT_{t - 1}|} \]
  Now consider an arbitrary ordering $o_1, o_2, \dots, o_{m}$ of $\OPT_{t - 1}$. We have
  \begin{align*}
    f(\OPT_{t - 1}) - f(\emptyset) &= \sum_{j = 1}^m (f(\{o_1, \dots, o_j\}) - f(\{o_1, \dots, o_{j - 1}\}))\\
    &\geq \sum_{j = 1}^m (f(\OPT_{t - 1}) - f(\OPT_{t - 1} \setminus \{o_j\}),
  \end{align*}
  where the inequality follows from submodularity: the marginal gain of $o_j$ on top of $\OPT_{t - 1} \setminus \{o_j\}$ is at most its marginal gain on top of $\{o_1, \dots, o_{j - 1}\} \subseteq \OPT_{t - 1} \setminus \{o_j\}$.

  Therefore
  \[ \Ex[f(\OPT_{t - 1}) - f(\OPT_{t - 1} \setminus \{o_t\})] \leq {f(\OPT_{t - 1}) \over |\OPT_{t - 1}|} \leq {f(\OPT) \over |\OPT_{t - 1}|}.\]
  To recap, we have shown that
  \[ \Ex[f(S_{t - 1} \cup \{e_t\}) - f(S_{t - 1})] \geq {c \cdot f(\OPT) \over |B_t|} \quad \text{and} \quad \Ex[f(\OPT_{t - 1}) - f(\OPT_{t - 1} \setminus \{o_t\})] \leq {f(\OPT) \over |\OPT_{t - 1}|}.\]
  Since $|B_t| = |\OPT_{t - 1}|$, we have
  \[ \Ex[f(S_{t - 1} \cup \{e_t\}) - f(S_{t - 1})] \geq c \cdot \Ex[f(\OPT_{t - 1}) - f(\OPT_{t - 1} \setminus \{o_t\})].\]
  Since the above inequality holds conditioned on every given event $R_{t - 1}$, it holds unconditionally, and the lemma follows.
\end{proof}

By combining Lemmas~\ref{lem:base-value} and \ref{lem:greedy-gain}, we obtain:

\begin{lemma}
\label{lem:sampling-greedy}
  If all of the buckets $B^{(j)}$ are good, the \textproc{LazySamplingGreedy} algorithm (Algorithm~\ref{alg:lazy-sampling-greedy}) returns a set $S \in \mathcal{I}$ with the following properties.
  \begin{enumerate}[$(1)$]
    \item $\max_{S' \colon S' \cup S \in \mathcal{I}} \sum_{e \in S'} f_S(e) \leq 4c M = O(1/\eps) f(\OPT)$.
    \item There is a random subset $\OPT' \subseteq \OPT$ depending on $S$ with the following properties: $S \cup \OPT' \in \mathcal{I}$ and $\Ex[f(\OPT')] \geq f(\OPT) - {1 \over c} \cdot \Ex[f(S)] \geq \left(1 - {1 \over c} \right) f(\OPT)$.
  \end{enumerate}
\end{lemma}
\begin{proof}
  The first property follows from Lemma~\ref{lem:base-value} and the stopping condition of \textproc{LazySamplingGreedy}. Indeed, let $v(e) = f_S(e)$ and $A \in \argmax_{S' \colon S' \cup S \in \mathcal{I}} v(S')$. When the algorithm stops, we have $w(B) \leq 4c M$. Additionally, by Lemma~\ref{lem:base-value}, we have
  \[ v(A) \leq w(A) \leq w(B) \leq 4c M = O(1/\eps) f(\OPT).\]
  Now consider the second property. Let $\OPT' = \OPT_k$. By Lemma~\ref{lem:greedy-gain},
  \[ \Ex[f(S)] = \sum_{t = 1}^k \Ex[f(S_t) - f(S_{t - 1})] \geq c \cdot \sum_{t = 1}^k \Ex[f(\OPT_{t - 1}) - f(\OPT_t)] = c (f(\OPT) - \Ex[f(\OPT')]).\]
  The second property now follows by rearranging the inequality above.
\end{proof}

By combining Lemmas~\ref{lem:continuous-greedy} and \ref{lem:sampling-greedy}, we obtain:

\begin{lemma}
\label{lem:combined-algo}
  The \textproc{ContinuousMatroid} algorithm (Algorithm~\ref{alg:matroid}) returns a solution $\mathbf{1}_S \vee x \in P(\mathcal{M})$ such that $F(\mathbf{1}_S \vee x) \geq (1 - 1/e - O(\eps)) f(\OPT)$ with constant probability.
\end{lemma}
\begin{proof}
  Note that, in order to apply Lemma~\ref{lem:continuous-greedy}, we need the following condition to hold:
  \[ \max_{S' \colon S' \cup S \in \mathcal{I}} \sum_{e \in S'} f_S(e) \leq c' f_S(\OPT''),\]
  where $\OPT'' \in \argmax_{S' \colon S' \cup S \in \mathcal{I}} f_S(S')$.

  Using Lemma~\ref{lem:sampling-greedy}, we can show that the above condition holds with constant probability as follows. Let $\OPT'$ be the set guaranteed by Lemma~\ref{lem:sampling-greedy}. We have $f_S(\OPT') \leq f_S(\OPT'')$ and $f(S \cup \OPT') \geq f(\OPT')$. Therefore
  \[f_S(\OPT'') \geq f_S(\OPT') \geq f(\OPT') - f(S).\]
 
 By Lemma~\ref{lem:sampling-greedy}, we have $\Ex[f(\OPT)-f(\OPT')] \le f(\OPT)/c$. Therefore, by the Markov inequality, with probability at least $2/3$, we have $f(\OPT)-f(\OPT') \le 3 f(\OPT)/c$. Consider two cases. First, if $f(S)\ge (1-1/e) f(\OPT)$ then the algorithm can simply return $S$. Second, if $f(S) < (1-1/e)f(\OPT)$ then $f_S(\OPT'')) \ge f(\OPT')-f(S) \ge (1/e - 3/c)f(\OPT)$. Therefore,
   \[\max_{S' \colon S' \cup S \in \mathcal{I}} \sum_{e \in S'} f_S(e) \le O(c f_S(\OPT'')) \leq c' f_S(\OPT''). \]
Thus the conditions of Lemma~\ref{lem:continuous-greedy} are satisfied and thus the continuous Greedy algorithm returns a solution $x \in P(\mathcal{M}/S)$ such that
   \begin{align*}
    F(\mathbf{1}_S \vee x) - f(S) 
    &\geq \left(1 - {1 \over e} - \eps \right)(f(\OPT') - f(S))\\
    &\geq \left(1 - {1 \over e} - \eps \right) \left(1 - {3 \over c}\right) f(\OPT) - f(S)\\
    &\geq \left(1 - {1 \over e} - 2\eps \right) f(\OPT) - f(S). 
   \end{align*}
\end{proof}

\section{Dynamic maximum weight base data structures}
\label{sec:data-structures}

In this section, we describe the data structures for maintaining a maximum weight base in a partition or a graphic matroid. The weights of the elements can only decrease and the data structure needs to support the operation \Call{UpdateBase}{} that decreases the weight of an element and it updates the base to a maximum weight base with respect to the new weights.

\subsection{The data structure for a graphic matroid}

For a graphic matroid, \Call{UpdateBase}{} is an update operation for a dynamic maximum weight spanning tree data structure. We can use the deterministic data structure by~\cite{HLT01}, which can handle each update in $O((\log k)^4)$ amortized time. Because the weights are only decreased, it suffices to use their decremental data structure instead of the fully dynamic one.

\subsection{The data structure for a partition matroid}

\begin{algorithm}
\caption{Update the maximum weight base for a partition matroid.}
\begin{algorithmic}[1]
\Procedure{UpdateBase}{$e, j, v(e)$}\Comment{Update $w(e)$ to $v(e)$ from $(1-\eps)^{j-1}M$}
    \State Let $i$ be the part containing $e$ ($e \in V_i$)
      \State Remove $e$ from $B$, $B^{(j)}$, and $V_i^{(j)}$
      \If{$v(e) > (1 - \eps)^N M$}
        \State Let $j'$ be such that $v(e) \in ((1 - \eps)^{j'} M, (1 - \eps)^{j' - 1} M]$
        \State $w(e) = (1 - \eps)^{j' - 1} M$
        \State Add $e$ to $V_i^{(j')}$
      \Else
        \State $w(e) = (1 - \eps)^N M$
      \EndIf
      \State Find $e'' \in \argmax_{e' \in V_i \setminus B} w(e')$
      \State Let $j''$ be such that $w(e'') = (1 - \eps)^{j'' - 1} M$
      \State Add $e''$ to $B$ and to $B^{(j'')}$ if $j'' \neq N + 1$
      \State $W = W - (1 - \eps)^{j - 1} M + (1 - \eps)^{j'' - 1} M$
\EndProcedure
\end{algorithmic}
\end{algorithm}

For a partition matroid, \Call{UpdateBase}{} simply needs to replaces $e$ in the base with an element $e'$ in its part with maximum weight. As we discuss below, the data structures representing the buckets allow us to implement the replacement in $O(1)$ amortized time.

In each of the parts $V_i$ of the matroid, we maintain the partition of elements into buckets where bucket $V_i^{(j)}$ contains all of the elements of $V_i$ with marginal values in the range $((1 - \eps)^j M, (1 - \eps)^{j - 1} M]$. The buckets of $B$ and the parts $V_i$ are represented as doubly linked lists (each item has a pointer to the previous and next items in the list). The list representing each of the buckets $V_i^{(j)}$ has the elements of $B^{(j)}$ at the front and there is a pointer to the first element of $V_i^{(j)} \setminus B^{(j)}$. This representation allows us to perform each of the following operations in constant time:
\begin{itemize}
\item We can remove an element from the list. 
\item We can add $e$ to $V_i^{(j)}$ as follows. If $e \in B$, we insert $e$ at the end of the prefix storing $B^{(j)}$ using the pointer to the first element after $B^{(j)}$. If $e \notin B$, we insert $e$ at the end of the list.
\item We can add $e$ to $B$ or $B^{(j)}$ by inserting it at the end of the list.
\end{itemize}
Additionally, we can find an element $e^* \in \argmax_{e \in V_i \setminus B} w(e)$ in constant amortized time as follows. A slower approach is to scan the buckets $V_i^{(j)}$ in increasing order (from $j = 1$ to $j = N$) to find the first bucket for which $V_i^{(j)} \setminus B^{(j)}$ is non-empty, and return the first element after $B^{(j)}$ (recall that $B^{(j)}$ is at the front of $V_i^{(j)}$ and we have a pointer to the first element of $V_i^{(j)}$ after $B^{(j)}$). This can be improved using the following observation: the maximum weight base contains from each part $V_i$ the elements of $V_i$ with maximum weight, and thus these elements appear in consecutive buckets. Therefore, for each part $V_i$, we can maintain a pointer to the first bucket $V_i^{(j)}$ that is partially full, i.e., $V_i^{(j)} \setminus B^{(j)}$ is non-empty.

\section{Swap rounding for a graphic matroid}
\label{sec:graphic-rounding}

In this section, we give a fast implementation of the swap rounding algorithm of Chekuri \etal \cite{Chekuri2010} for a matroid polytope, shown in Algorithm~\ref{alg:swap-round}. The rounding algorithm takes as input a point $x$ represented as a convex combination $x = \sum_{i = 1}^t \beta_i \mathbf{1}_{B_i}$, where $B_i$ is a base of $\mathcal{M}$ (a base is an independent set with maximum cardinality).    

\begin{algorithm}
\caption{The swap rounding algorithm for a matroid polytope from \cite{Chekuri2010}}
\label{alg:swap-round}
\begin{algorithmic}[1]
\Procedure{SwapRound}{$x = \sum_{i = 1}^t \beta_i \mathbf{1}_{B_i}$}
\State $C_1 \gets B_1$
\State $\gamma_1 \gets \beta_1$
\For{$i = 1$ to $t - 1$}
\State $C_{i + 1} \gets \Call{MergeBases}{\gamma_i, C_i, \beta_{i + 1}, B_{i + 1}}$
\State $\gamma_{i + 1} \gets \gamma_i + \beta_{i + 1}$
\EndFor
\State Return $C_t$
\EndProcedure
\Procedure{MergeBases}{$\beta_1, B_1, \beta_2, B_2$}
\While{$B_1 \neq B_2$}
\State Pick $i \in B_1 \setminus B_2$ and find $j \in B_2 \setminus B_1$ such that $B_1 - i + j \in \mathcal{I}$ and $B_2 - j + i \in \mathcal{I}$
\State With probability ${\beta_1 \over \beta_1 + \beta_2}$
\State \quad $B_2 \gets B_2 - j + i$
\State Else
\State \quad $B_1 \gets B_1 - i + j$
\EndWhile
\State Return $B_1$
\EndProcedure
\end{algorithmic}
\end{algorithm}

In this section, we show that, when $\mathcal{M}$ is a graphical matroid, we can implement the swap rounding algorithm so that it runs in $O(nt \log^2 n)$ time, where $n$ is the rank of the matroid (the number of vertices in the graph) and $t$ is the number of bases in the convex combination of $x$. The overall running time of \Call{SwapRound}{} is the time of $t - 1$ calls to \Call{MergeBase}{}. In the following, we show that we can implement the \Call{MergeBase}{} subroutine so that it runs in $O(n \log^2{n})$ time.

\medskip
{\bf Data structure for representing a spanning tree.}
In a graphical matroid, each base is a spanning tree. To obtain a fast implementation of \Call{MergeBase}{}, we will represent each spanning tree using the following data structure. Consider a spanning tree $T$. For each vertex $v$ of $T$, we use a gadget that is a red-black tree with $deg(v)$ copies of $v$, each copy is connected to one edge of $v$ in the original tree (see Figure~\ref{fig:tree-gadget}). 

\begin{figure}[h]
\begin{center}
\includegraphics[scale=0.5]{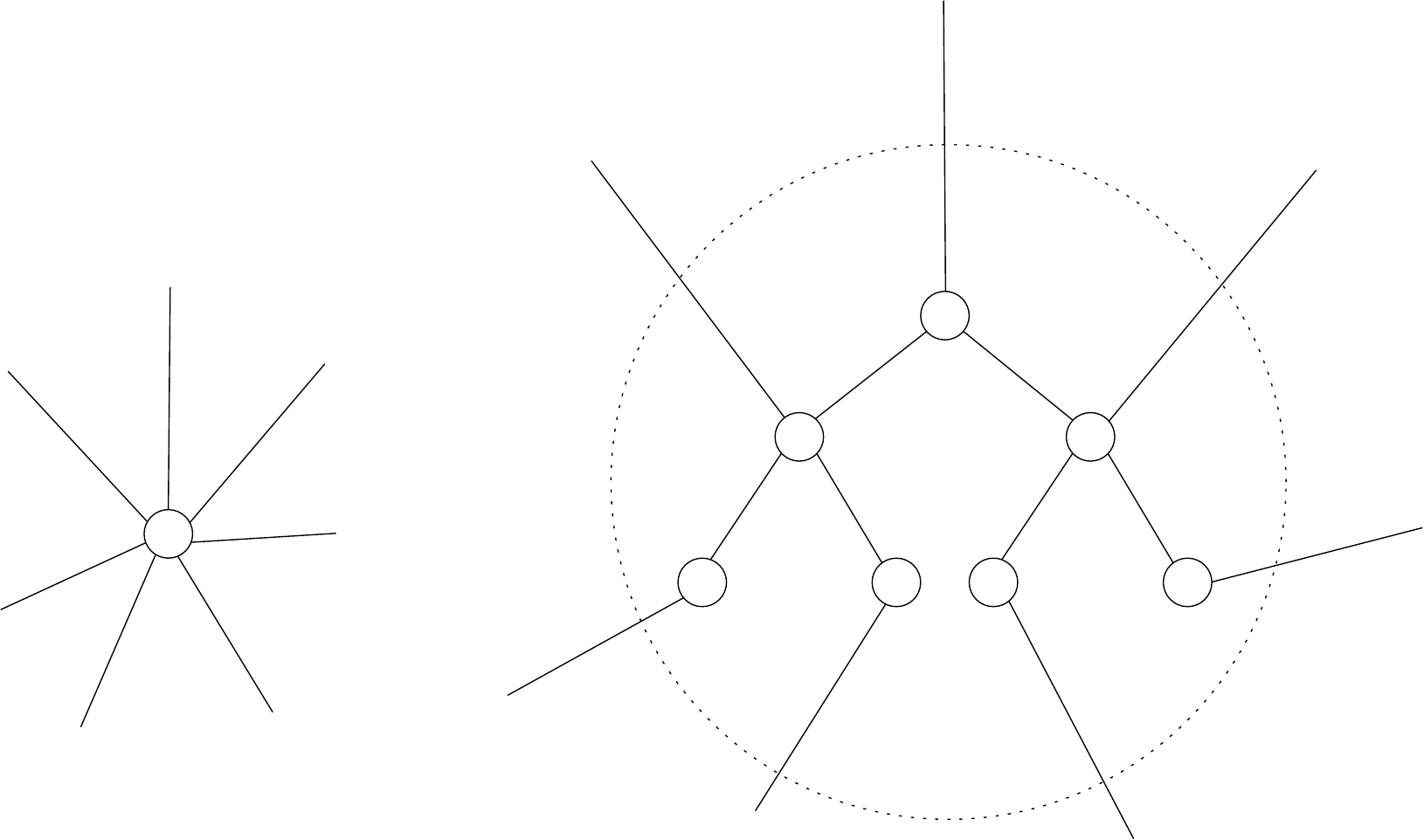}
\end{center}
\caption{Representing a vertex of degree $d$ using a red-black tree with $d$ nodes.}
\label{fig:tree-gadget}
\end{figure}

Let $T'$ be the tree obtained from $T$ by replacing each node by the gadget described above. We store the adjacency list of $T'$ and an Eulerian tour of $T'$ represented using the Euler tour data structure of Henzinger and King \cite{HK95}. The Euler tour data structure represents any forest in such a way that each of the following operations takes $O(\log{n})$ amortized time per operation:  
\begin{itemize}
\item $\Call{link}{T_1, T_2, e}$: connect two trees $T_1$ and $T_2$ of the forest into a single tree $T = T_1 \cup T_2 \cup \{e\}$, where $e$ is an edge connecting $T_1$ and $T_2$;
\item $\Call{cut}{T, e}$: split $T$ into two trees by removing the edge $e$;
\item $\Call{same-tree}{u, v}$: return whether $u$ and $v$ are in the same tree of the forest. 
\end{itemize}

To summarize, the data structure representing a spanning tree $T$ has the following components:
\begin{itemize}
\item The red-black trees representing the gadgets of each node (Figure~\ref{fig:tree-gadget}).
\item The adjacency list of the tree $T'$ obtained from $T$ by expanding each node by its gadget.
\item The Euler tour data structure representing the expanded tree $T'$.
\end{itemize}

The reason for using a red-black tree for the gadget is that each insertion and deletion involves $O(1)$ rotations, leading to only $O(1)$ link-cut operations to update the Euler tour data structure.

We now describe how to implement the operations needed for merging two spanning trees. At the beginning of the $\Call{MergeBase}{T_1, T_2, \beta_1, \beta_2}$, we contract in $T_1$ and $T_2$ each edge that is in both trees. This can be accomplished similarly to the \Call{swap-and-contract}{} operation described below (there is no swap, only the contraction). We can upper bound the running time of all of the contractions by $O(n \log^2{n})$ using the same analysis as the one given below for \Call{swap-and-contract}{}. We then proceed with the swap and update operations.

The data structure will support the following two operations for two spanning trees $T_1$ and $T_2$ such that $T_1 \neq T_2$:

\begin{itemize}
\item $\Call{find-swap}{T_1, T_2}$ finds a pair $(e, f)$ of edges such that $e \in T_1 \setminus T_2$ and $f \in T_2$, and $T_1 \setminus \{e\} \cup \{f\}$ and $T_2 \setminus \{f\} \cup \{e\}$ are spanning trees.

\item $\Call{swap-and-contract}{T, T', e, e'}$ implements the update $T \gets T \setminus \{e\} \cup \{e'\}$ as follows: update the data structure representing $T$ to remove $e$ and add $e'$, and update the data structures representing $T$ and $T'$ to contract $e'$ in both trees.
\end{itemize}

The following claims upper bound the time to perform the \Call{swap-and-contract}{} operation.

\begin{claim}
We can implement the update $T \gets T \setminus \{e\} \cup \{e'\}$ in $O(\log{n})$ amortized time. 
\end{claim}
\begin{proof}
Removing edge $e = (u,v)$ requires a cut operation on the Eulerian tour data structure representing $T$ at the location of $e$, and a deletion of $e$ from the red-black trees of vertices $u$ and $v$. Each insertion and deletion in a red-black tree has $O(\log{n})$ time and $O(1)$ rotations. Therefore the total time for removing $e$ is $O(\log n)$.  

Adding edge $e' = (x, y)$ requires an insertion of $e'$ to the red-black trees of vertices $x$ and $y$, and a link operation on the Eulerian tour at the location of $e'$. Each of these operations takes $O(\log n)$ time so the total time for adding $e'$ is $O(\log n)$.
\end{proof}

\begin{claim}
We can contract an edge $e = (u, v)$ in $T$ in $O(deg(v) \log{n})$ amortized time.
\end{claim}
\begin{proof}
To contract an edge $(u, v)$, we meld the red-black tree of $v$ into the red-black tree of $u$: we break up the red-black tree of $v$ into individual nodes, and we insert them one by one into the red-black tree of $u$. This requires $deg(v) - 1$ insert operations in the red-black tree of $u$, and each insert operation has an $O(\log{n})$ amortized cost. Therefore the amortized cost is $O(deg(v) \log{n})$.
\end{proof}

We next discuss how to implement the \Call{find-swap}{} operation. The key idea is to use an edge $e$ adjacent to a leaf in $T_1$ as opposed to an arbitrary edge. It is straightforward to maintain the degrees of the vertices in the tree in a such a way that we can retrieve a leaf node in O(1) time. Let $v$ be a leaf node in $T_1$ and let $e = (v, u)$ be its incident edge. We can find the edge $f \in T_2$ (the swapping partner of the leaf edge $e$) as follows. 

\begin{figure}[h]
\begin{center}
\begin{subfigure}{0.25\linewidth}
\includegraphics[scale=0.5]{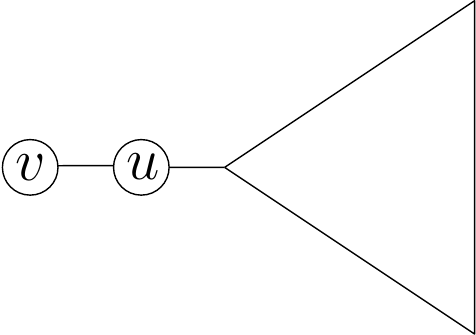}
\subcaption{$T_1$}
\end{subfigure}
\begin{subfigure}{0.3\linewidth}
\includegraphics[scale=0.5]{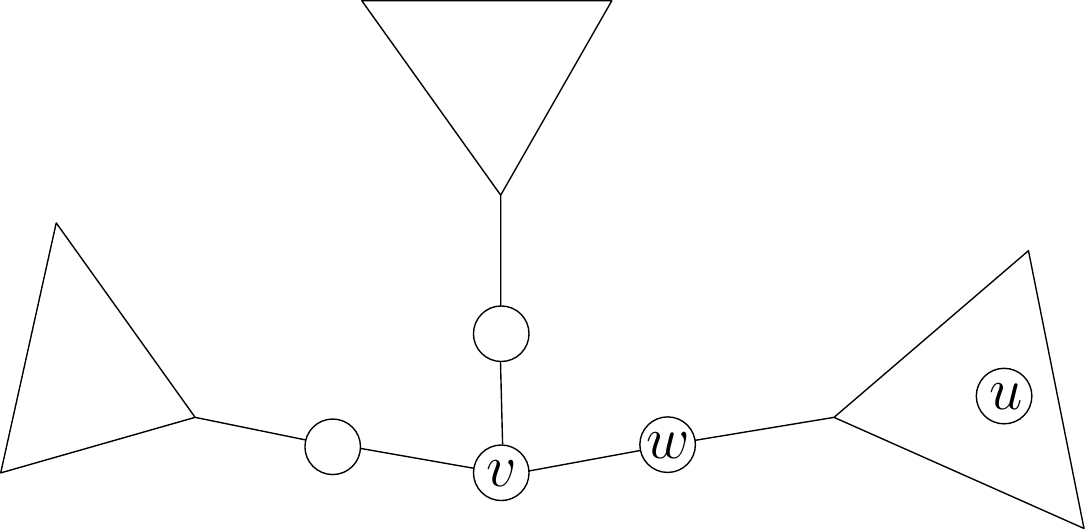}
\subcaption{$T_2$}
\end{subfigure}
\end{center}
\caption{$v$ is a leaf in $T_1$. The swapping partner of the edge $(v,u)$ in $T_1$ is the unique edge $(v,w)$ in $T_2$ such that $w$ is on the path between $u$ and $v$ in $T_2$}
\label{fig:swapping-partner}
\end{figure}
\begin{figure}[h]
\begin{center}
\includegraphics[scale=0.5]{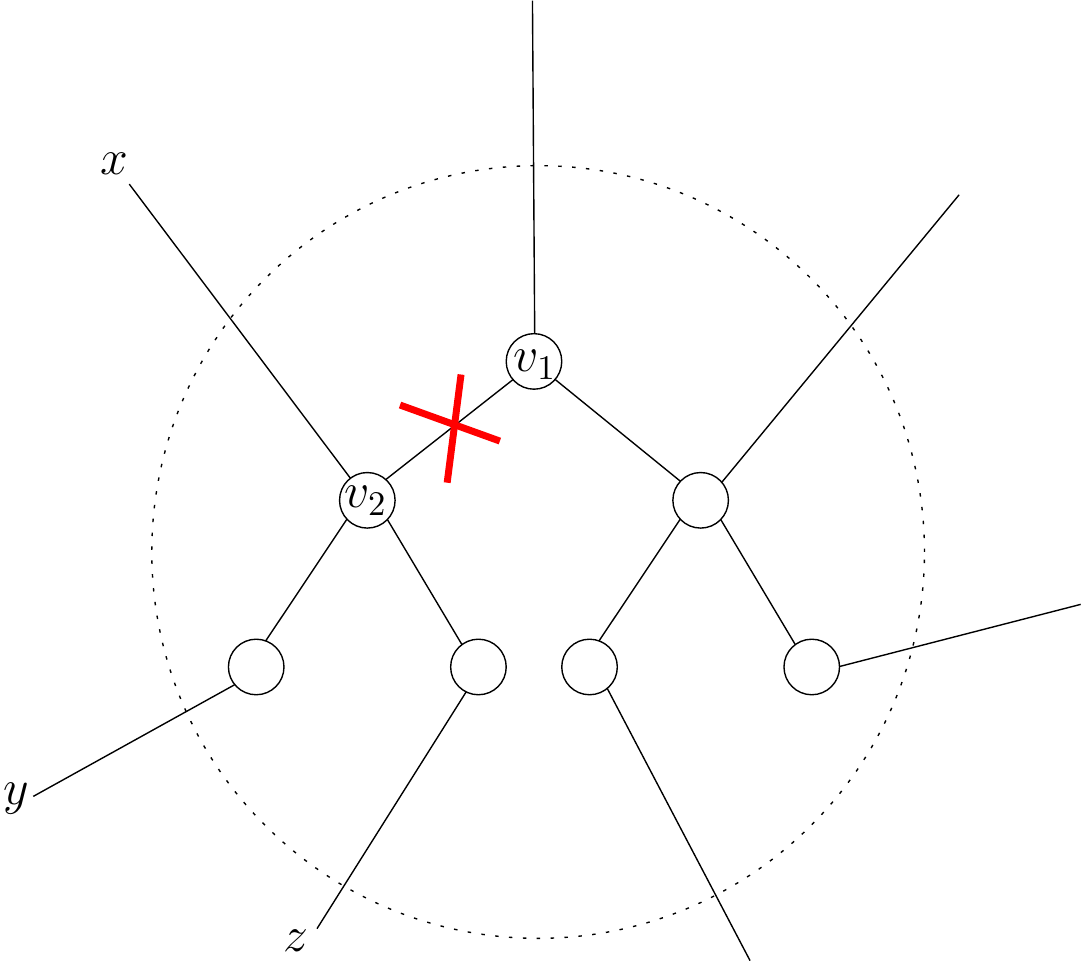}
\end{center}
\caption{To check if $v$ is connected to $u$ via one of the edges on the left half of $v$'s red-black tree, we can remove the edge between $v_1$ and $v_2$ and check if $v_1$ and $u$ are still connected.}
\label{fig:find-w}
\end{figure}

Notice that the swapping partner of $(v, u)$ must be an edge $(v, w)$ where $w$ is on the path between $v$ and $u$ in $T_2$ (Figure~\ref{fig:swapping-partner}). We can find the vertex $w$ by performing a binary search in the red-black tree of $v$ in $T_2$ as follows. In the following, by the tree $T_2$ we mean the expanded tree where each vertex is replaced by its red-black tree gadget. We start at the root of the red-black tree of $v$ in $T_2$. We want to find $v$ by searching in the appropriate subtree of the red black tree (either the subtree rooted at the left child or the subtree comprised of the root together with the subtree rooted at the right child). To this end, we remove the edge connecting the root and its left child from $T_2$ using a cut operation (Figure~\ref{fig:find-w}). This splits the red-black tree of $v$ into two parts and also splits $T_2$ into the two subtrees. Each subtree of $T_2$ has a copy of $v$, where $v$ in the first subtree has $i$ edges of $v$ and the copy of $v$ in the second subtree has the remaining $deg(v)-i$ edges. In each subtree, we check whether $v$ is connected to $u$ in that subtree by performing a $\Call{same-tree}{v, u}$ operation on the subtree. After performing the check, we know in which of the two subtrees $v$ is connected to $u$: if it is the first subtree, $(v,w)$ is one of the first $i$ edges adjacent to $v$; otherwise, $(v,w)$ is one of the remaining $deg(v)-i$. We add the edge that we removed back to $T_2$ using a link operation, and recursively search in the subtree in which $v$ and $u$ are connected. Since the red-black tree has depth $O(\log{n})$, we can find $w$ using $O(\log{n})$ operations. Thus the overall amortized cost for finding $w$ is $O(\log^2{n})$.

\medskip
{\bf Running time of \Call{MergeBase}{}.} In a single iteration of \Call{MergeBase}{}, we perform a \Call{find-swap}{} operation and a \Call{swap-and-contract}{} operation. The number of \Call{find-swap}{} operations is at most $n$ and each operation takes $O(\log^2{n})$ amortized time, and thus the overall running time of the \Call{find-swap}{} operations is $O(n \log^2{n})$. By using the potential $O(\log{n}) \sum_{x\in V} deg(x)\log deg(x)$, we can show that the overall running time of the \Call{swap-and-contract}{} operations is $O(n \log^2{n})$ as well. Thus the running time of \Call{MergeBase}{} is $O(n \log^2{n})$.

\medskip
{\bf Running time of \Call{SwapRounding}{}.} To round a convex combination of $t$ spanning trees, we perform $t - 1$ calls to \Call{MergeBase}{}, and thus the running time is $O(t n \log^2{n})$.

\bibliographystyle{abbrv}
\bibliography{submodular}
\end{document}